\documentclass[]{article}
\usepackage[english]{babel}
\usepackage{graphicx}
\usepackage{amsmath,amssymb,amsthm}
\usepackage{url}
\usepackage{mathtools}
\usepackage{mathrsfs}
\usepackage[font=small,skip=0pt]{caption}
\usepackage{subcaption}
\usepackage{pgfplots}
\usepackage[normalem]{ulem}
\usepackage{algorithm}
\usepackage{algpseudocode}

\newcommand{\setfont}{\mathbb}
\newcommand{\real}{\ensuremath{\setfont{R}}}

\renewcommand{\natural}{\ensuremath{\setfont{N}}}

\newtheorem{theorem}{Theorem}
\newtheorem{proposition}{Proposition}

\newcommand{\Gr}{\mathcal{G}} 
\newcommand{\Kg}{\mathcal{K}}
\newcommand{\lm}{\underline{\lambda}}
\newcommand{\lp}{\bar{\lambda}}

\begin{document}

\title{Accelerating Consensus by Spectral Clustering and Polynomial Filters}

\author{Simon~Apers \thanks{S. Apers is with the SYSTeMS group, Ghent University, Belgium; e-mail: simon.apers@ugent.be} 
        \,and~Alain~Sarlette \thanks{A. Sarlette is with the QUANTIC project team, INRIA Rocquencourt, France and the SYSTeMS group, Ghent University, Belgium; e-mail: alain.sarlette@inria.fr}}

\maketitle

\begin{abstract}
It is known that polynomial filtering can accelerate the convergence towards average consensus on an undirected network. In this paper the gain of a second-order filtering is investigated. A set of graphs is determined for which consensus can be attained in finite time, and a preconditioner is proposed to adapt the undirected weights of any given graph to achieve fastest convergence with the polynomial filter. The corresponding cost function differs from the traditional spectral gap, as it favors grouping the eigenvalues in two clusters. A possible loss of robustness of the polynomial filter is also highlighted.
\end{abstract}
  

\section{Consensus acceleration} \label{sec:introduction}

Since their introduction in \cite{tsitsiklis1984}, (discrete-time) consensus algorithms have attracted almost as much attention as their dual, fast mixing Markov chains \cite{boyd2004,chen1999}. Improving the convergence speed of this basic building block for e.g.~distributed computation and sensor fusion has been a major focus. For synchronized fixed networks, one can optimize the weights on the links \cite{boyd2004}, add local memory \cite{muthukrishnan1998}, or introduce time-varying filters \cite{montijano2013,hendrickx2014}. The present paper establishes the benefit of combining polynomial filtering with optimization of link weights.

Consider an undirected and connected graph $\mathcal{G}(\mathcal{V},\mathcal{E})$ with $N$ nodes $\in \mathcal{V}$ and $M$ edges $\in \mathcal{E}$. Denote the node states as $x = (x_1\,, x_2,...,\,x_N) \in \mathbb{R}^N$. The basic linear consensus dynamics on $\mathcal{G}$ is
\begin{equation} \label{eq:stdcons}
x_i(k\!+\!1) = x_i(k)+w\smashoperator{\sum_{(i,j)\in \mathcal{E}}}L_{ij}(x_j(k)\text{-}x_i(k))
\end{equation}
with $k\in\natural$, $w \in \mathbf{R}$ some gain and $L$ a symmetric matrix of edge weights called the Laplacian. We can rewrite \eqref{eq:stdcons} as 
$$x(k\!+\!1)= (I-wL)\,x(k) = P_w\,x(k) \, ,$$ 
with $I$ the identity matrix.
\emph{Consensus} is the stationary state satisfying
    \begin{equation*}
     x_1=x_2=...=x_N =: c \;\; \text{or equivalently} \;\; P_w\, x = x
    \end{equation*}
(if the graph is connected); and \emph{average consensus} requires the consensus value $c$ to satisfy $c=\frac{1}{N}\sum_{i=1}^Nx_i(0)$.
The convergence speed of \eqref{eq:stdcons} towards consensus is governed by $P_w$'s second largest eigenvalue modulus (SLEM)
\begin{equation*}
\mu(P_w) = \max\{|x^TP_wx|\,: x^T x \leq 1,\,{\textstyle \sum_{k=1}^N} x_k=0\}
\end{equation*}
or the ``spectral gap'' $1-\mu(P_w)$. For a given graph, the fastest convergence is obtained by choosing the edge weights to maximize the spectral gap. This is a convex problem \cite{boyd2004}, whose solution we call the Fastest Single-Step Convergence network (FSSC)\footnote{Boyd et al. study the optimization for Markov chains, with the constraint $[P_w]_{i,j}>0$ for all $i,j$. For consensus that constraint also favors robustness to network changes. In the present paper we drop it, for a fairer comparison with other algorithms; the optimization remains a convex problem.}.
In case of a fixed $L$, with knowledge of bounds $\lm<\lp$ on its eigenvalues, the optimal $\mu$ is obtained by choosing $w$ such that $\mu = (1\text{-}w\lm) = - (1\text{-}w\lp)$.\\

\textit{Accelerated consensus} denotes the expansion of node and communication features, under the same graph constraint, to improve convergence speed. 

In \cite{muthukrishnan1998} and later papers, \emph{memory slots} are added at each node to get
\begin{equation}\label{eq:memslot}
x(k\!+\!1) = (I-w_1 L)x(k) + w_2 (x(k\!-\!1)-x(k))
\end{equation}
with some fixed $w_1,w_2 \in \mathbb{R}$. Extending the memory registers does not improve convergence speed if only bounds $\lm < \lp$ on the eigenvalues of a fixed $L$ are known~\cite{sarlette2014}.

Another approach is \emph{polynomial filtering}, where a time-varying choice of $P_w$ accelerates convergence \cite{montijano2013,hendrickx2014}. The idea is that if $L$ has eigenvectors $\tilde{x}_i$ with eigenvalues $\lambda_i$, then after $t$ steps of $P_w(k)=I-w(k) L$, each eigenvector has been multiplied by $p_t(\lambda_i) = \Pi_{k=1}^t (1-w(k) \lambda_i)$. Through choosing a set of $w(k)$, $p_t$ can be any polynomial of order $t$ satisfying $p_t(0)=1$; whereas constant $w$ restricts to $p_t(\lambda_i) = (1-w\lambda_i)^t$. 
Choosing $w(k) = 1/\lambda_k$ would imply finite-time convergence. The latter requires not only to implement a high-order polynomial filter (see Section \ref{sec:implementation}), but also to know the eigenvalues exactly. In the opposite case where we only know that the eigenvalues lie in a bounded interval, we must apply the following result.

\begin{proposition}\label{prop1}
If the eigenvalues of $L$ span the whole interval $\lambda_i \in [\lm,\lp]$, then the optimal memory slot dynamics \eqref{eq:memslot} is at least as fast as \eqref{eq:stdcons} with any time-varying $w$ (i.e.~any order of polynomial filtering).
\end{proposition}
\begin{proof}
\cite{montijano2013} gives the optimal polynomial filter and its worst-case long-term convergence speed: for $w$ such that $\mu = (1\text{-}w\lm) = - (1\text{-}w\lp)$ we have 
$\max_{\left| 1-w \lambda \right|< \mu} [p_t(\lambda)] = 1/\left| T_t(1/\mu) \right| \,$
with $p_t$ rescaled from a $t$-order Chebyshev polynomial $T_t$. A function plot readily shows that
\begin{equation}\label{eq:montconv}
\frac{1}{\left| T_t(1/\mu) \right|^{1/t}} \geq \frac{1}{\mu}-\sqrt{\frac{1}{\mu^2}-1}
\end{equation}
for all $\mu \in (0,1)$, with equality as $t\rightarrow +\infty$. The right side of \eqref{eq:montconv} is precisely the worst-case convergence speed with optimal \eqref{eq:memslot}, see \cite{muthukrishnan1998,sarlette2014}.
\end{proof}

Whether a synchronous variation of $w$ or a local memory slot is a more demanding resource, is application dependent. One local memory allows to evaluate $T_t$ with a stable recurrence \cite{muthukrishnan1998}. Our new observation with Proposition \ref{prop1} is that with respect to this implementation, \eqref{eq:memslot} is superior.\\

However, in the following, we show how much more can be gained by the polynomial filter when (somewhat) more is known about $L$. We restrict our investigation to a two-step alternating scheme:
\begin{equation}\label{eq:ourman}
x(k+2)=(I\text{-}w_2 L) (I\text{-}w_1 L)\, x(k) \; ,
\end{equation}
performing second-order polynomial filtering on the $L$-spectrum (or equivalently on the $P$-spectrum). 

Note that acceleration schemes have also been proposed for Markov chains (see e.g.~\cite{chen1999} and related work). However, the most prominent schemes explicitly build on full knowledge of a given $\mathcal{G}$, which differs from our setting.

The remainder of the paper is organized as follows. In Section \ref{sec:considerations} we compute the optimal second-order polynomial and the associated gain in convergence speed. We also mention some special graphs that allow 2-step consensus with possible symmetry-breaking. In Section \ref{sec:preconditioning} we consider the optimization of graph weights towards polynomial filtering. Section \ref{sec:robustness} discusses the implementation of the acceleration scheme and its robustness properties.
  
  
\section{Optimizing the polynomial filter} \label{sec:considerations}

We now investigate the form of the optimal second-order polynomial filter \eqref{eq:ourman} for a given graph $\mathcal{G}$, about which we possibly know more than just the spectral gap. Before the rescaling of $L$ in \eqref{eq:ourman}, we can assume without loss of generality that $P := I-L$ is centered, i.e.~$-\mu(P) x^T x \leq x^T P x \leq \mu(P) x^T x$ for all $x$ with $\sum_{i=1}^N x_i = 0$, with each equality achieved for some $x$, and where $\mu(P)$ is the SLEM of $P$. In this situation, the optimal polynomial filter only depends on $P$'s SLEM and \emph{Smallest Eigenvalue Modulus} $\sigma(P)$ (SEM).

Explicitly, rewrite \eqref{eq:ourman} as
\begin{equation}\label{eq:ourman2}
x(k+2)=\frac{(P-z_1)(P-z_2)}{(1-z_1)(1-z_2)}\, x(k) =: p_2(P)\, x(k)
\end{equation}
and denote by $\{ \lambda_i : i=1,2,...,N \}$ the eigenvalues of $P$, with $\lambda_1 = 1$ corresponding to the consensus eigenspace. Note that the filter $p_2$ can be any second-order polynomial restricted to $p_2(1)=1$. We define the optimal exponential convergence rate as:
\begin{equation*}
\mu_2(P) \stackrel{\vartriangle}{=} \min_{z_1,z_2} \left[ \max_{i>1} |p_2(\lambda_i)| \right]\, .
\end{equation*}

\begin{theorem}\label{th:main}
Consider a connected, undirected graph $\Gr$ with given (centered) weight matrix $P$. The optimal convergence rate attainable through $p_2$-acceleration \eqref{eq:ourman2} is given by
	\begin{equation} \label{eq:mu2}
	  \mu_2(P) = \frac{\mu(P)^2-\sigma(P)^2}{2-\mu(P)^2-\sigma(P)^2}
	\end{equation}
and obtained with the unique polynomial
	\begin{equation}\label{eq:pe2}
	  p_2(P)=\frac{P^2-\mu(P)^2/2-\sigma(P)^2/2}{1-\mu(P)^2/2-\sigma(P)^2/2}.	
	\end{equation}      
\end{theorem}
\begin{proof}
For any $[|a|,|b|] \subseteq [0,1]$, the second-order polynomial $p_2(x)$ for which $p_2(1)=1$ and which minimizes $\max_{x \in \{a,b,-b\}}\left| p_2(x) \right|$ is determined by $\, p_2(b) = p_2(-b) = -p_2(a) = -p_2(-a) \,$. Replacing $|a|,|b|$ by particular eigenvalues $\sigma(P),\mu(P)$, this polynomial bounds the best possible performance. Moreover, for this same polynomial we have $|p_2(x)| \leq |p_2(b)|$ for any $|x| \in [a,b]$, so it actually gives the best performance, independent of the other eigenvalues of $P$.
\end{proof}

The convergence rate $\mu_2(P)$ must be compared with $\mu(P)^2$, the convergence over \emph{two steps} of the standard consensus algorithm \eqref{eq:stdcons}. The resulting improvement by $p_2$-acceleration is illustrated on 
Figure \ref{fig:p2}.

\begin{figure}[htb]
\centering
\includegraphics[width=0.75\columnwidth]{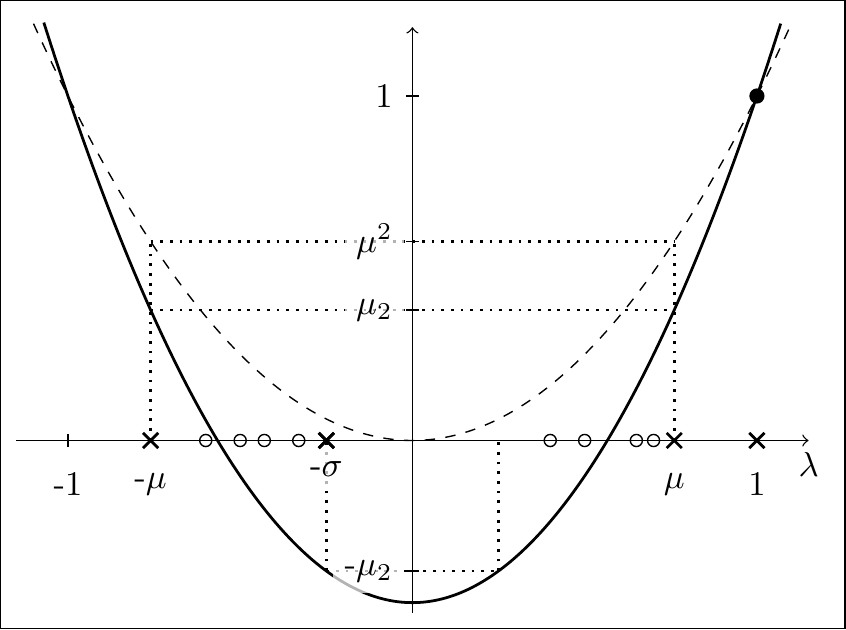}\vspace{2mm}
\caption{Comparison of SLEM after 2 steps of standard consensus ($\mu^2$) and convergence rate under optimal $p_2$-acceleration ($\mu_2$), for some arbitrary $P$. The critical eigenvalues, determining the convergence rate, are distinguished ($\times$) from the other ones ($\circ$). The polynomials $y = \lambda^2$ (dotted lines, standard consensus) and $y=p_2(\lambda)$ (plain line, polynomial filter) illustrate the acceleration mechanism.}
\label{fig:p2}
\end{figure} 


\subsection{Favorable graphs and symmetry breaking} \label{sec:favorablegraphs}

In light of the previous result, a given weighted  graph $\Gr$ will allow substantial acceleration through second-order filtering if its (e.g.~Laplacian) eigenvalues are clustered in two sets whose width is small compared to the distance between the sets. This situation easily generalizes to the case of $n$ clusters for $n$-order polynomial filters, at least conceptually; explicit forms generalizing \eqref{eq:mu2},\eqref{eq:pe2} may be more difficult to obtain.

This observation motivates a different way of optimizing network links, if some design freedom is available. Before turning to an investigation of such optimization, we briefly consider the special cases where \emph{finite-time convergence to consensus is achieved} in two steps.\\

The study of finite-time consensus has its own line of work, to which the present results connect through \cite{hendrickx2014,sandryhaila2014} and the like. Obviously, if a $P$ matrix in our framework has only two distinct eigenvalues $\lambda_2=-\lambda_3$ besides the invariant space with $\lambda_1=1$, then $\sigma(P) = \mu(P)$ and \eqref{eq:mu2},\eqref{eq:pe2} imply perfect consensus after one application of \eqref{eq:ourman}. Graphs $\Gr$ for which $P_1$, $P_2$ can be selected such that $P_1 \, P_2$ features deadbeat convergence are characterized in \cite{hendrickx2014}. The polynomial filter is restricted to $P_2 = cI + d P_1$, with $c,d \in \mathbb{R}$. It is presently unclear how much more restrictive this is, but the following results nevertheless carry over from \cite{hendrickx2014}.

\begin{proposition}\label{prop2}
$\bullet$ If 2-step convergence can be reached on a graph $\Gr$ then its diameter $d(\Gr)\leq 2$.\newline
$\bullet$ $d(\Gr)\leq 2$ is not sufficient for 2-step convergence.\newline
$\bullet$ If in addition $\Gr$ is distance-regular, then 2-step convergence can be attained.
\end{proposition}

The condition $d(\Gr)\leq 2$ turns out to be sufficient for $N\leq 5$ at least, showing that $5$ out of the $6$ four-node graphs (resp.~15 ouf of 21 five-node graphs) converge in two steps with $p_2$-acceleration; while for \eqref{eq:stdcons} only the complete graph converges in finite time. Lists of graphs with only two different nonzero eigenvalues can be found in the literature, see e.g.~\cite[Table 14.2, Table 14.4, Chapter 15.2]{brouwer2011} assuming uniform weights.

Examining the complete bipartite graphs $K(\ell,m)$ between sets of $\ell$ and $m$ nodes allows to illustrate interesting finite-time convergence properties. Both $\Kg(m,m)$ and the star graph $\Kg(1,m)$ allow finite-time convergence using $p_2$ with uniform link weights, although only $\Kg(m,m)$ is distance-regular. For the other $\Kg(\ell,m)$ cases, uniform weights lead to 3 distinct nonzero eigenvalues. Regarding $\mu^2$ (FSSC) this choice is optimal, see~\cite{boyd2004,boyd2009}. Regarding $\mu_2$, a nonuniform weight selection might further accelerate the convergence.
In particular, a \emph{symmetry breaking} on $\Kg(2,m)$ does allow finite-time convergence. Indeed, with asymmetric weights $p\neq q$ (see figure \ref{fig:K2mtadpole}), the nonzero eigenvalues are proportional to:
      \begin{equation*}
       p+q,\;\; \tfrac{(m+1)(p+q)}{2} 
	\pm\tfrac{\sqrt{(m+1)^2(p+q)^2-4pqm(m+2)}}{2}.
      \end{equation*}
By choosing $q=\frac{1}{2}(m\pm\sqrt{m^2-4})p$, this set reduces to two distinct values and $\mu_2(P)=0$. \emph{This possible benefit of symmetry breaking contrasts with standard consensus \eqref{eq:stdcons}, for which \cite{boyd2009} shows that keeping the edge-transitivity symmetry in the weights does lead to the FSSC.}
    \begin{figure}[ht]
      \centering
      \includegraphics[width=0.55\columnwidth]{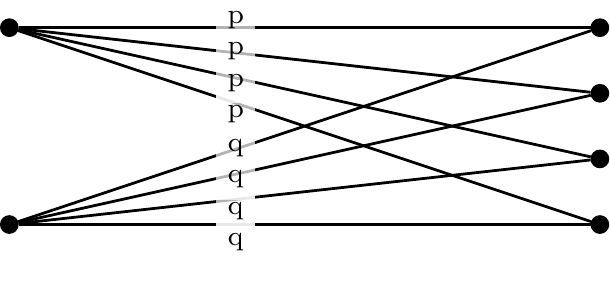}
     \caption{The complete bipartite graph $\Kg(2,4)$ with weights $p$ and $q$.}
     \label{fig:K2mtadpole}
    \end{figure}
 
Optimization of the edge weights towards $p_2$-acceleration, for a given graph structure, is further considered in the following section.   


\section{Preconditioning graphs: spectral clustering} \label{sec:preconditioning}

  Theorem \ref{th:main} describes the convergence rate obtained with an optimal polynomial filter for a given weighted graph (spectrum). This section treats the optimization of the edge weights in compliance with the graph connectivity, which we call \emph{preconditioning} the graph towards $p_2$-acceleration. This is in the same spirit as the FSSC. However, optimal \emph{polynomial} acceleration favors an eigenvalue spectrum in distinct \emph{clusters}. A positive point is that Theorem \ref{th:main} gives an explicit expression for the optimal polynomial, so we can efficiently concentrate on optimizing $P$. For a given $\mathcal{G}(\mathcal{V},\mathcal{E})$, we denote by $P^*$ the FSSC i.e.~the optimal choice of weights for \eqref{eq:stdcons}, by $P^*_{p2}$ the optimal weights for $p_2$ acceleration.

\subsection{Bound on fastest convergence rate}

  We first get a bound on the possible acceleration. Given $\Gr$, denote $\Gr^2$ its \emph{square graph}, in which two nodes are linked by an edge iff in $\Gr$ they are linked by a path of $\leq 2$ edges. We denote an arbitrary weight matrix on this power graph by $P_{\Gr^2}$.
    \begin{proposition}\label{lowerbndr}
      The fastest convergence rate attainable on $\Gr$ using $p_2$-acceleration is bounded by the SLEM of the Fastest Single-Step Consensus on $\Gr^2$, i.e.~
      $$\mu_2(P^*_{p_2}) \geq \mu(P^*_{\Gr^2}) \; $$
    \end{proposition}
    \begin{proof}
      For $p_2(P)$ we effectively tune the same edge weights as for $P_{\Gr^2}$, except that in the latter the individual weights are unconstrained while in the former they are interdependent through the polynomial structure.
    \end{proof}
  This bound is not always tight: if the diameter $d(\Gr)\leq 2$ then $\Gr^2$ is completely connected and $\mu(P^*_{\Gr^2})=0$. The second item of Proposition \ref{prop2} however states that this does not imply $\mu_2(P^*_{p_2}) = 0$.
   
\subsection{Numerical optimization} \label{sec:conv}

The FSSC choice of $P^*$ does not necessarily minimize $\mu_2$, see e.g.~the $\mathcal{K}(2,m)$-graph (section \ref{sec:favorablegraphs}); an alternative optimization is required.\\
Expressed as a function of the individual edge weights, $\mu_2(P)$ as defined in Theorem \ref{th:main} is not a convex function, unlike $\mu(P)$ for the FSSC. Given the convexity of $\mu(\cdot)$ as a function of its argument and inspired by Proposition \ref{lowerbndr}, we try to reformulate the problem as minimizing $\mu(P_{\Gr^2})$ over all $P_{\Gr^2}$ which can be represented as $p_2(P)$. The resulting relaxation however is not expected to improve convexity as the set of all $p_2(P)$ is a non-convex subset of the set of all $P_{\Gr^2}$ (see Proposition \ref{prop:nonconv} in appendix). So to have a convex problem we indeed need to relax the subspace to all $P_{\Gr^2}$. The solution of the relaxed problem would then have to be reprojected into the set of all $p_2(P)$. And for this, it is well-known that variations on matrix elements give little insight on the induced variations in eigenvalues.

Thus optimal preconditioning towards $p_2$-acceleration seems significantly more difficult numerically than FSSC, unless a better reformulation is found. In the meantime, we have performed a numerical optimization by \emph{gradient descent}, minimizing the explicit function $\mu_2$ as a function of the edge weights. We claim by no means that this is the best strategy in terms of complexity or results, at this point it is just a feasible one to evaluate the potential of preconditioning for $p_2$-acceleration. As the used formulation is non-convex, the method converges to a local minimum, not guaranteeing the global optimum. The bound $\mu(\Gr^2)$ from Proposition \ref{lowerbndr} can give an indication on the quality of the obtained $P$ matrix.
  
Figure \ref{fig:egvclust} shows the spectrum of an Erd\H{o}s-R\'enyi graph with 20 nodes, 96 edges and diameter 2. Its weights have been optimized respectively for FSSC ($P^*$) and for acceleration preconditioning ($P^*_{p2}$, hopefully). The difference is graphically striking as the preconditioning brings $\mu_2(P^*) = 0.0626$ down to $\mu_2(P^*_{p2}) = 0.0088$. This is further to be compared to $\mu^2(P^*) = 0.1181$ for the FSSC without $p_2$ acceleration and to the lower bound $\mu(P^*_{\Gr^2}) = 0$.
    
   \begin{figure}
      \centering
      \includegraphics[width=0.8\columnwidth,clip=true,trim = 6.05cm 11.6cm 6.4cm 11.6cm]{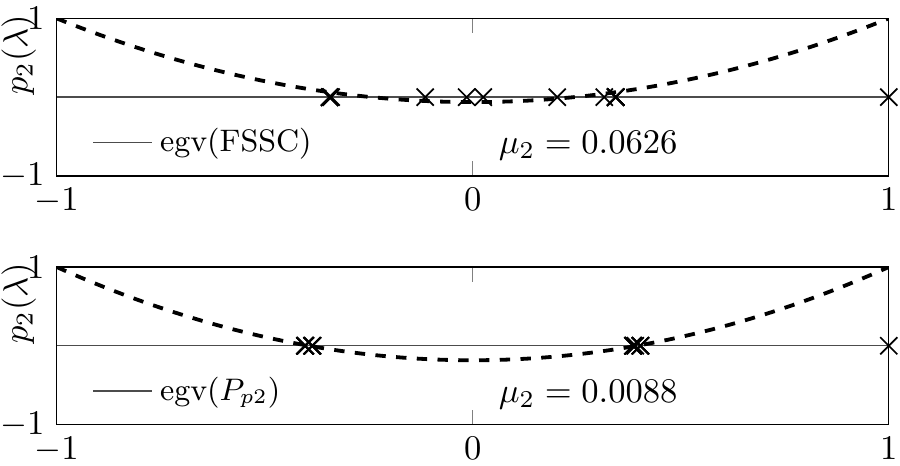}  
      \caption{Spectral clustering performed on an Erd\H{o}s-R\'enyi graph $\mathcal{G}(20,9)$. The plot shows the spectrum of both the FSSC $P^*$ and the optimized $P_{p2}$, as well as their optimal $p_2$-polynomials and corresponding $\mu_2$ values.}
      \label{fig:egvclust}   
   \end{figure}  
   
Beneficial clustering of the eigenvalues around two polynomial zeros $\pm z$ as on Fig.~\ref{fig:egvclust} is not always possible. It typically deteriorates with decreasing number of edges per node, thus as the amount of degrees of freedom available for optimizing a constant number of eigenvalues decreases. That this is not a general rule follows from the perfect clustering of star graphs (section \ref{sec:favorablegraphs}). 
We observe the following behavior in simulations.  
\begin{itemize}    
    \item The behavior of the preconditioner and acceleration was examined on a large set of Erd\H{o}s-R\'enyi graphs. A clear trend appears when examining the average results as a function of the density, i.e. the number of edges compared to the complete graph, as shown in figure \ref{fig:clustdens}. When less than about 30\% of the edges are present w.r.t.~the complete graph, the most significant acceleration is obtained by just taking the graph optimized for the FSSC, which is easily computable, and applying the optimal polynomial filter to it instead of first-order consensus; trials to further adjust the weights towards faster convergence do not really pay off (in average). For densities higher than about 30~\%, the preconditioner starts to significantly pay off, with gains up to orders of magnitude.    
    \item This behavior and the previously mentioned lower bound of Proposition \ref{lowerbndr} might suggest that a sort of phase transition should appear as a function of the \emph{diameter} of the graph. To investigate this we have partitioned the set of Erd\H{o}s-R\'enyi graphs of a given density, as a function of their diameter. Except for the graphs with finite-time convergence at diameter 2, the acceleration ratio achieved by the preconditioner with respect to the FSSC appears to have \emph{the same} distribution on these subsets. Moreover, the lower bound of Proposition \ref{lowerbndr} was reached on several graphs of diameter 3,4 and 5. This seems to indicate that the graph diameter is not a limiting factor for the preconditioner. Surprisingly, it appears that the role of the graph diameter in possible convergence speed is an open question in the literature also for the FSSC/FMMC problem.    
    \item The optimal $P$ matrix sometimes contains negative elements, also on off-diagonal entries. This indicates that \emph{repulsion} between certain nodes can accelerate consensus with $p_2$. Such negative entries however can be undesirable for robustness, in which case one can easily exclude them in the optimization process.   
    \item Unsurprisingly, a spectrum well clustered with the preconditioner is often highly degenerate in $\pm \mu$ and $\pm \sigma$. This is similar to the degeneracy found at $\pm \mu$ in the spectrum of the FSSC.
  \end{itemize}
  
    \begin{figure}
      \centering
      \includegraphics[width=0.7\columnwidth]{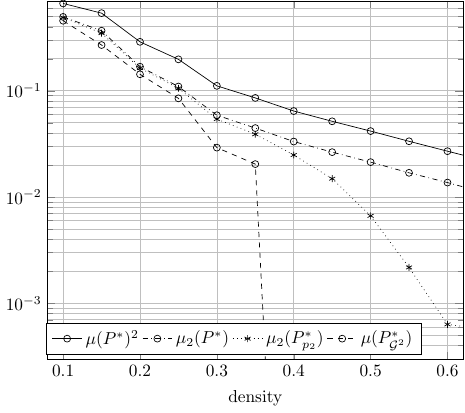}  
      \caption{Performance of clustering on an Erd\H{o}s-R\'enyi graph with increasing density. Depicted are $\mu(P^*)^2$ using the FSSC during two steps, its $p_2$-acceleration $\mu_2(P^*)$, the $p_2$-acceleration with optimized graph $\mu_2(P^*_{p2})$ (hopefully), and the lower bound $\mu(P^*_{\Gr^2})$.}
      \label{fig:clustdens}   
    \end{figure}


\section{Robustness} \label{sec:robustness}
    
  
\subsection{Implementation in alternating steps} \label{sec:implementation}
 
Implementation of the acceleration scheme \eqref{eq:pe2} requires to sequentially apply two matrices $P_-$ and $P_+$, such that 
    \begin{align*}
      P_-P_+ &= \dfrac{P^2-z^2}{1-z^2}, & \text{with } z^2 &= \dfrac{\mu(P)^2+\sigma(P)^2}{2}.
    \end{align*}     
A linear implementation \eqref{eq:ourman} compatible with one-step communication links for fixed $P$, necessarily takes the form
\begin{align} \label{eq:impl}
P_- &= a\frac{P-z}{1-z}, & P_+ &= \frac{1}{a}\frac{P+z}{1+z}, & a&\in\real. 
\end{align}
The effect of $P_-$ and $P_+$ on a given eigenvector is obtained just by replacing $P$ in \eqref{eq:impl} by the corresponding eigenvalue. We immediately see that the eigenvalue $1$ of $P$, corresponding to the consensus eigenvector, will be multiplied alternatively by $a$ and by $1/a$. In this sense $a>1$ (or $1/a>1$) thus implies an ``unstable'' step with $P_-$ (or with $P_+$) for the consensus eigenvector. The consensus value is kept \emph{at each step} only if $a=1$, else it is recovered every second step. However, taking $a\neq 1$ can be interesting regarding the remaining eigenvalues of $P$. We say that the iteration towards consensus involves stable steps if both $\mu(P_+) < 1$ and $\mu(P_-) < 1$, i.e.~they are stable on the eigenspace orthogonal to the consensus eigenvector associated to the trivial eigenvalue $1$ of $P$.

\begin{proposition}\label{prop:impl}
For given $z$: 
    $\bullet$ The matrices $P_+$ and $P_-$ with $a=1$ are both stable for consensus if and only if $\mu(P) < 1-2z$.
\newline
    $\bullet$ The iteration towards consensus involves stable steps for a proper choice of $a\neq 1$ if and only if $\mu(P) < 1-z$.
\newline
For given $\mu(P)$ it is always possible to restrict $z$ such that these stability conditions are satisfied, possibly by taking it smaller than the optimal acceleration value $z^2 = (\mu(P)^2 + \sigma(P)^2)/2$.
\end{proposition}
\begin{proof}
Choosing $a=1$ we get, by linearity of both $P_+$ and $P_-$ in $P$ and hence in its eigenvalues: $\max\{|\mu(P_+)|,|\mu(P_-)|\} = (\mu(P)+z)/(1-z)$. The latter is smaller than 1 if and only if $\mu(P) < 1-2z$; while the eigenvalue $1$ of $P$ remains unchanged in $P_+$ and $P_-$ when $a=1$. Choosing $a=\sqrt{1-z}/\sqrt{1+z}$ minimizes $\max\{|\mu(P_+)|,|\mu(P_-)|\}$, at the value $(\mu(P)+z)/\sqrt{1-z^2}$. The latter is smaller than 1 if and only if $\mu(P)<1-z$.
\end{proof}

Figure \ref{fig:SLEMconv} illustrates the second case of Proposition \ref{prop:impl}. The stability of individual steps can be relevant if we cannot ensure that all intended steps will be applied, e.g.~due to synchronization issues. For standard consensus with $P$ this implies no problem, just skipped steps. But in accelerated consensus with $a=1$, if by chance $P_-$ is applied more frequently than $P_+$ in the situation of Fig.~\ref{fig:SLEMconv}, then one mode increases in an unstable way. We can prevent this risk of instability with $P_-^{st}$, which takes $a\neq 1$. In this case, more frequent applications of $P_-$ will just change the consensus \emph{value} to something (unstably) different from the average of initial values, but it will not prevent the agents from converging to consensus. The preferable tradeoff depends on the practical situation.
\begin{figure}[htb]
      \centering
      \includegraphics[width=0.9\columnwidth,clip=true,trim=5.4cm 19.7cm 4.4cm 4.2cm]{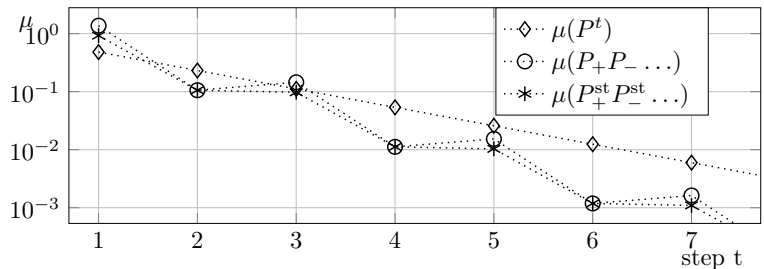}
      \caption{Convergence of the individual steps when using standard consensus ($P$), optimal $p_2$ acceleration with $a=1$ (periodic repetition of $P_+P_-\ldots$) and with $a=\sqrt{1-z}/\sqrt{1+z}$ (periodic repetition of $P_+^{\text{st}}P_-^{\text{st}}\ldots$), for the graph with optimized weights shown in figure \ref{fig:robustness}.}
      \label{fig:SLEMconv}
\end{figure}
 

\subsection{Robustness to link failure} \label{sec:linkfailure}

Consider a consensus scheme tailored to an initial undirected graph $\Gr$ with weight matrix $P$, featuring positive weight on each link. A sudden edge failure leads to a modified graph $\mathcal{G}'$ with weights $P'$. The standard consensus dynamics will remain stable under this failure, as $P'$ remains doubly-stochastic, see e.g.~\cite{brouwer2011}. Link failure affects the $p_2$-accelerated dynamics as follows.
 
\begin{proposition} \label{prop:linkfailure}
$\bullet$ For certain graphs, the optimal $p_2$ filter for a weight matrix $P$ associated to positive edge weights can become unstable with weights $P'$ in which a link has failed (permanent link failure), and also if the link fails one step out of two (resonant link failure).\newline
$\bullet$ Restricting $z \leq 1/\sqrt{2}$ (resp.~$z \leq (1-\mu(P))/2$) in the polynomial filter ensures stability under permanent (resp. resonant) link failures from a $P$ associated to positive edge weights, as for standard consensus.\newline
$\bullet$ Any $P$ matrix which features some negative edge weights can become unstable under specific link failures, both for standard consensus and with $p_2$ acceleration.\newline
$\bullet$ Robustness to link failure is independent of the choice of $a$ in \eqref{eq:impl}.
\end{proposition}
\begin{proof}
The last point results from the fact that even if links fail at some times, in absence of other casualties, the scheme keeps alternating the two steps of \eqref{eq:impl}, possibly with different $P$ matrices but still with the $a$ factors canceling.
\newline The third point is trivial if we consider the case where all links with positive edge weights (attraction between agents) fail, while all links with negative edge weights (repulsion between agents) remain. Indeed, this leaves only repulsive dynamics and we can write $P' = I + L'$ where $L'$ is a Laplacian with non-negative eigenvalues -- i.e. some eigenvalues of $P'$ will necessarily be larger than 1. Let us now turn to the case of $P$ restricted to positive edge weights.
\newline Under permanent link failure, the polynomial filter is applied to a different set of eigenvalues $\in [-1,1]$, all closer to $1$ than the original ones \cite{brouwer2011}. If $\vert p_2(\lambda) \vert \leq 1$ for all $\lambda \in [-1,1]$ there is no risk of instability. However, if $p_2(0) < -1$, an eigenvalue of $P$ might become close to zero after link failure and lead to instability, see Figure \ref{fig:robustness}. This leads to the condition $z \leq 1/\sqrt{2}$.
\newline Under resonant link failure, when $P'_+$ corresponds to all links failing and $P'_-$ to no link failure, the eigenvector of $\mu(P)$ is multiplied by $\frac{-\mu(P)-z}{1-z}$ over two time steps. This is the worst case: just consider the worst graph achievable with link failures separately for each step. The condition expresses $\left| \frac{-\mu(P)-z}{1-z} \right| < 1$. \end{proof}

The instability under link failure is caused by a potentially unstable region on the polynomial, characteristic of highly clustered spectra. Figure \ref{fig:robustness} shows a 5-node graph whose clustered spectrum has an unstable region in the center, i.e.~ $|p_2|>1$ on some interval inside $[-1,1]$. Indeed, when any of the dashed edges fails, one of the eigenvalues hops into this region. For both permanent and resonant failure the scheme will turn unstable. Constraining the polynomial according to Proposition \ref{prop:linkfailure} restores robustness of the scheme but lowers its acceleration.

   \begin{figure}[htb]
      \centering
      \includegraphics[width=0.9\columnwidth,clip=true,trim=4.5cm 20.5cm 5.5cm 4.4cm]{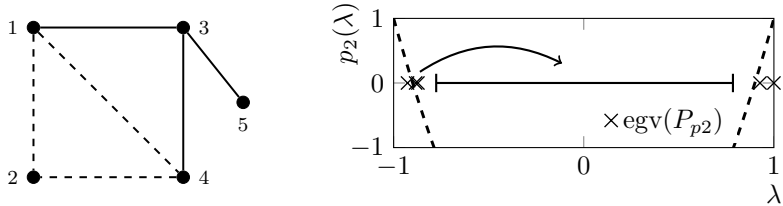}
      \caption{The graph on 5 nodes defined by the union of full and dashed edges on the left, is preconditioned for $p_2$ acceleration with weights {\small $(1,2),(2,4)=0.628$, $(1,4)=0.605$, $(1,3),(3,4)=0.045$, $(3,5)=0.926$}. Its (highly clustered) spectrum and optimal $p_2$ are shown on the right. Failure of any of the dashed edges leads to instability as an eigenvalue moves into the unstable region of $p_2$, as indicated by the arrow on the spectrum.}
      \label{fig:robustness} 
  \end{figure}


\section{Discussion} \label{sec:discussion}

\subsection{Higher Order Polynomial Filtering}
  While this paper focuses on second-order acceleration, several results can be extended to arbitrary order polynomials $p_M$ while others remain open.
  \begin{itemize}
    \item The closed form for the optimal polynomial --- if there actually is a unique one --- still needs to be investigated. It will certainly exploit more graph information than $\mu(P)$ and $\sigma(P)$.
    \item Graphs allowing finite-time convergence can be investigated, with obvious improvement as $M$ tends to $N$, the graph order. Note that still for $p_2$, we do not have an exact answer to the question, among others due to the necessity to consider beneficial symmetry-breakings.
    \item The bound of Proposition \ref{lowerbndr} holds verbatim in the form $\mu_M(P^*_{p_M})\geq \mu(P^*_{\Gr^M})$.
    \item The non-convexity of Proposition \ref{prop:nonconv} holds, replacing each edge in the example of the proof by a path of appropriate length.
    \item For the gradient descent and associated investigation, a closed form for the optimal polynomial would be welcome, else the polynomial parameters can be part of the optimization variables (see e.g. \cite{sandryhaila2014}).
    \item The robustness discussion remains qualitatively the same.
  \end{itemize}
However, relevant situations for \emph{practical implementation} of consensus with high-order polynomial filters would probably be the first question to consider.
  
\subsection{Conclusions}
In this paper we have characterized the possibilities to accelerate linear consensus by second-order polynomial filtering as proposed in \cite{montijano2013}. We have observed that this strategy is beaten by an acceleration based on local memory slots if only an upper and a lower bound are known on the graph spectrum. However when more is known about the graph spectrum, performance can be improved significantly. For a graph with fixed weights the optimal filter and its convergence rate were derived exactly. A preconditioner is proposed which optimizes the edge weights of a given graph, clustering its eigenvalues towards better polynomial acceleration. Unlike for standard consensus this optimization appears to be non-convex. Significant payoffs are obtained especially for graphs with high edge density.
  
A few academic questions remain open. A particular one is whether the possibility of achieving consensus in $k$ steps with a time-varying weight matrix on a fixed graph \cite{hendrickx2014} implies that a $k$-order polynomial filter can also achieve finite-time consensus. A more general question, for which we have been surprised to find no answer in the literature even regarding standard consensus algorithms, is how the diameter of a graph might bound the best achievable convergence rate with optimized edge weights.
  
In a broader scope, we notice that an approach similar to polynomial filtering has been proposed a few decades ago to control LTI systems using periodic memoryless output feedback \cite{aeyels1992}. They show that introducing periodically varying feedback can widen the eigenvalue assignment possibilities. We anticipate that those accelerations based on additional memory or time-dependent actions could also be linked to the memory effects and parallel actions present in quantum random walks \cite{aharonov1993}. We are currently working on formalizing this link in the emerging field of quantum systems engineering.

\appendix

\begin{proposition} \label{prop:nonconv}
The set of all $p_2(P)$ is a non-convex subset of the set of all $P_{\Gr^2}$.
\end{proposition}
\begin{proof}
Consider $\Kg(1,4)$, the star graph on 5 nodes (see fig.~\ref{fig:convcross}).
Let a given $P'$ have equal positive weights only on the edges (1,2) and (1,3), and $P''$ have the same positive weights only on (1,4) and (1,5). Then for $p'_2$, $p''_2$ some second-order polynomials, $p'_2(P')$ and $p''_2(P'')$ correspond to positive weights on respectively (1,2),(1,3),(2,3) and (1,4),(1,5),(4,5). Their convex combination $1/2\, p''_2(P'') + 1/2\, p'_2(P')$ cannot be generated by any $p_2(P)$, because (strong) weights on (2,3), (4,5) require positive weights on all 4 edges of the star, which in turn unavoidably imply (non-negligible) positive weights on (2,4),(2,5),(3,4),(3,5) in any second-order $p_2(P)$.
\end{proof}

\begin{figure}[htb]
\centering
\includegraphics[width=0.7\columnwidth]{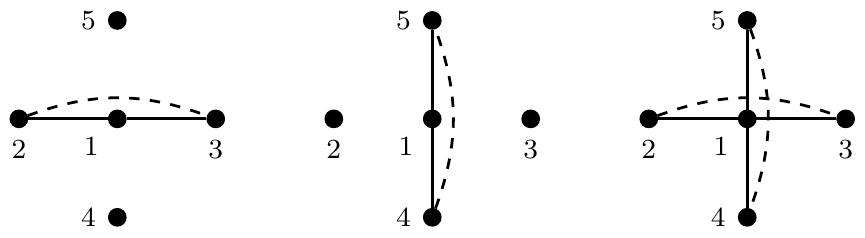}
\caption{\textit{Left:} $P'$, dashed line added for $p_2'(P')$. \textit{Center:} idem for $P''$ and $p_2''(P'')$. \textit{Right:} Convex combination.}
\label{fig:convcross}
\end{figure}

\section*{Acknowledgment}
The authors thank F.Ticozzi for discussions that have drawn their interest to this topic and B.Gerencs\'er for fruitful discussions mainly about Section \ref{sec:favorablegraphs} and Proposition \ref{prop:nonconv}. The authors are partially supported by the Interuniversity Attraction Poles program DYSCO, funded by the Belgian Science Policy Office.

  \bibliography{p2acc-paper}

\end{document}